\documentclass[11pt]{article}

\usepackage[utf8]{inputenc} 
\usepackage[T1]{fontenc}    
\usepackage{url}            
\usepackage{booktabs}       
\usepackage{amsfonts}       
\usepackage{nicefrac}       
\usepackage{microtype}      
\usepackage{xcolor}         
\usepackage{booktabs}
\usepackage{float}
\usepackage{tabularx}
\usepackage{array}
\usepackage{comment}
\usepackage{natbib}
\usepackage[margin = 1in]{geometry}

\usepackage[table]{xcolor}

\usepackage{amsmath}        
\usepackage{amssymb}        
\usepackage{amsthm}         
\usepackage{mathtools}      
\usepackage{bm}             
\usepackage{algorithm}
\usepackage{algpseudocode}

\usepackage{subcaption}

\usepackage{placeins}
\usepackage{float}
\usepackage{enumitem}
\usepackage{tikz,lipsum,lmodern}
\usepackage[most]{tcolorbox}

\theoremstyle{plain}
\newtheorem{theorem}{Theorem}

\newtheorem{proposition}[theorem]{Proposition}

\theoremstyle{definition}

\theoremstyle{remark}

\newcommand{\poubelle}[1]{}

\title{Directionally Correct Views in Portfolio Optimization}

\author{
\bfseries Ojas Kumar$^1$ \ \ \ Sreejit Dey$^1$ \ \ \ Sarvesh R. Iyer$^1$ \ \ \ Sandeep Juneja$^1$ \\ Ashoka University \footnote{\{ojas.kumar / sreejit.dey / sarvesh.iyer / sandeep.juneja\}@ashoka.edu.in} 
}

\renewcommand{\addcontentsline}[3]{}

\usepackage{todonotes}
\usepackage{titlesec}
\begin{document}

\maketitle

\begin{abstract}
Portfolio optimization techniques often use historical data to estimate parameters for Markowitz-type or ML models. Effective allocation requires a way to align the available data with conditions expected over the investment horizon. Inspired by the Black–Litterman framework, we ask whether incorporating expert views leads to better portfolio allocations. In a simple Markowitz setting, we show that even a small amount of directionally correct information can materially improve investment decisions. We also discuss how expert views can be incorporated in end-to-end portfolio models using entropy-based ideas. Somewhat surprisingly, our empirical results indicate that simpler Markowitz-type model perform as well as an end-to-end model, once directionally correct experts are  incorporated.
\end{abstract}

\section{Introduction}

Portfolio design is a historically rich field, combining simple financial insights with complex mathematical models. A seminal contribution to this field was made by Markowitz \cite{Markowitz1952portfolio}, who first introduced a quantitative framework in financial investment eventually leading to the Capital Asset Pricing Model (CAPM) \cite{sharpe1964capital,lintner1965valuation,mossin1966equilibrium}. A major drawback of the theory lay in the inability to predict the mean $\boldsymbol{\mu}$ and covariance $\boldsymbol{\Sigma}$ of future asset returns, and the extreme sensitivity of the optimal allocation to any estimation errors \cite{best1991sensitivity,chopra1993effect, jorion1986bayes,ledoit2004honey,choueifaty2008toward,maillard2010properties, jagannathan2003risk}.

Black and Litterman \cite{black1992global} in their pioneering work developed a Bayesian framework to arrive at forward looking mean--variance estimates by incorporating qualitative \emph{expert views} into a quantitative framework. This work inspired many subsequent works in portfolio optimization as well as in option pricing, e.g.,  \cite{buchen1996maximum, avellaneda2001weighted, stutzer1996simple, meucci2005beyond, meucci2008fully, pezier2007global, qian2001conditional, simonian2014copula}. Precisely what quantitative factors determine the effectiveness of expert views is a question of importance. More recently, data-driven approaches have sought to avoid estimating
\(\boldsymbol{\mu}\) and \(\boldsymbol{\Sigma}\) \cite{zhang2020deep,zhang2021universal}, raising the question of how forward-looking expert views can be incorporated directly into an end-to-end model.

\textbf{Overview.} In this article, we provide quantitative estimates showing that even a
modest amount of \textit{directionally correct} information in expert views can substantially improve
portfolio allocations, and we develop a framework for incorporating such views into both
Markowitz and end-to-end models. \S~\ref{sec:srunderev} shows that the Sharpe ratio increases
monotonically, and is concave near the optimum, as the directionally correct view
strengthens, and \S~\ref{expmare} shows this holds robustly under small errors in the view.
\S~\ref{sec:end-to-end} extends this notion to an end-to-end, ML-based portfolio model and
finds empirically that it performs comparably to Markowitz once directionally correct views
are incorporated. We use the Sharpe ratio throughout as our measure of portfolio performance
\cite{sharpe1966mutual,sharpe1994ratio}.

\section{Sharpe Ratio under Directionally Correct Views}\label{sec:srunderev}

Under the standard Markowitz framework, we consider $n$ assets with random return vector $\mathbf r=(r_1,\ldots,r_n)^\top$. Its mean and covariance are denoted by
$\boldsymbol{\mu}:=\mathbb E[\mathbf r]$, and $\boldsymbol{\Sigma}:=\operatorname{Cov}(\mathbf r)$.
A portfolio $\mathbf w\in\mathbb R^n$ has return
$r_{\mathbf w}=\mathbf w^\top\mathbf r$, with
$
    \mathbb E[r_{\mathbf w}]
    =\mathbf w^\top\boldsymbol{\mu},$ and $\operatorname{Var}(r_{\mathbf w})
    =\mathbf w^\top\boldsymbol{\Sigma}\mathbf w$. We measure the performance of a portfolio by the widely-used Sharpe ratio,
\begin{equation}\label{eq:SR}
    \operatorname{SR}(\mathbf w)
    =
    \frac{
        \mathbf w^\top\boldsymbol{\mu}-r_f
    }{
        \sqrt{\mathbf w^\top\boldsymbol{\Sigma}\mathbf w}
    }.
\end{equation}
Here, $r_f$ is the risk-free rate, which we will assume to be $0$ without loss of generality. We refer to the portfolio maximizing \eqref{eq:SR} as the
\emph{maximum-Sharpe portfolio.} When short-selling is allowed and $\boldsymbol{\Sigma}$ is invertible, the maximum-Sharpe portfolio has
weights 
$\mathbf w^\ast
    \propto
    \boldsymbol{\Sigma}^{-1}\boldsymbol{\mu}$, where $\mathbf 1$ denotes the vector of ones; see \citep{merton1972analytic}.

Suppose an expert's prediction of the future mean and covariance are $\boldsymbol{\nu}$ and $\boldsymbol{\Omega}\succ0$ respectively. To determine how this portfolio performs under the distribution that actually governs returns, we calculate the realized Sharpe ratio and maximum possible Sharpe ratio as follows:
\begin{equation}
\label{eq:misspecified-sharpe}
    \operatorname{S}_{\boldsymbol{\mu},\boldsymbol{\Sigma}}
    (\boldsymbol{\nu},\boldsymbol{\Omega})
    :=
    \frac{
        \boldsymbol{\mu}^{\top}
        \boldsymbol{\Omega}^{-1}\boldsymbol{\nu}
    }{
        \sqrt{
            \boldsymbol{\nu}^{\top}
            \boldsymbol{\Omega}^{-1}
            \boldsymbol{\Sigma}
            \boldsymbol{\Omega}^{-1}
            \boldsymbol{\nu}
        }
    }, \quad   S^\star_{\boldsymbol{\mu},\boldsymbol{\Sigma}}
    := \operatorname{S}_{\boldsymbol{\mu},\boldsymbol{\Sigma}}
    (\boldsymbol{\mu},\boldsymbol{\Sigma}) = 
    \sqrt{
        \boldsymbol{\mu}^{\top}
        \boldsymbol{\Sigma}^{-1}
        \boldsymbol{\mu}
    }.
\end{equation}

\begin{proposition}[Sharpe loss as an angle]
\label{prop:sharpe-angle}
Assume $\boldsymbol{\mu},\boldsymbol{\nu}\neq0$ and
$\boldsymbol{\Sigma},\boldsymbol{\Omega}\succ0$. There exists a
unique angle $\theta_{\boldsymbol{\mu},\boldsymbol{\Sigma}}
(\boldsymbol{\nu},\boldsymbol{\Omega})\in[0,\pi]$
such that
$ 
\cos
\theta_{\boldsymbol{\mu},\boldsymbol{\Sigma}}
(\boldsymbol{\nu},\boldsymbol{\Omega})
=
\frac{
    \boldsymbol{\mu}^{\top}\boldsymbol{\Omega}^{-1}\boldsymbol{\nu}
}{
    \sqrt{
        \boldsymbol{\mu}^{\top}
        \boldsymbol{\Sigma}^{-1}\boldsymbol{\mu}
    }
    \sqrt{
        \boldsymbol{\nu}^{\top}
        \boldsymbol{\Omega}^{-1}
        \boldsymbol{\Sigma}
        \boldsymbol{\Omega}^{-1}\boldsymbol{\nu}
    }
}.
$ 
Moreover, we have
$ 
S_{\boldsymbol{\mu},\boldsymbol{\Sigma}}
(\boldsymbol{\nu},\boldsymbol{\Omega})
=
S^\star_{\boldsymbol{\mu},\boldsymbol{\Sigma}}
\cos
\theta_{\boldsymbol{\mu},\boldsymbol{\Sigma}}
(\boldsymbol{\nu},\boldsymbol{\Omega}),
$ so the Sharpe ratio is maximized when $\boldsymbol{\Omega}^{-1}\boldsymbol{\nu}
=
c\,\boldsymbol{\Sigma}^{-1}\boldsymbol{\mu}$
for some $c>0$.
\end{proposition}

We say that an expert view is \emph{directionally correct} if, when $\boldsymbol{\mu}$ is the true return mean and $\boldsymbol{\nu}$ is our present forecast, the expert's prediction is of the form $\boldsymbol{\nu}_\lambda
:=
(1-\lambda)\boldsymbol{\nu}
+\lambda\boldsymbol{\mu}$,
for some $\lambda \in [0,1]$. For the next result, we assume that our original prediction, the expert view, as well as the real covariance matrix are the same (namely $\boldsymbol{\Sigma}$). We define $A:=\langle\boldsymbol{\nu},\boldsymbol{\nu}\rangle_{\boldsymbol{\Sigma}^{-1}},$ and $B:=\langle\boldsymbol{\nu},\boldsymbol{\mu}-\boldsymbol{\nu}\rangle_{\boldsymbol{\Sigma}^{-1}}$, where $\langle\boldsymbol{x},\boldsymbol{y}\rangle_{\boldsymbol{\Sigma}^{-1}}
:=\boldsymbol{x}^{\top}\boldsymbol{\Sigma}^{-1}\boldsymbol{y}$.

\begin{theorem}[Monotonicity and concavity in $\lambda$]
\label{thm:mean-improvement}
Let $S(\lambda)$ denote the Sharpe ratio of the portfolio constructed
by the mean estimate $\boldsymbol{\nu}_\lambda$, evaluated under the
true return distribution. That is, $S(\lambda):=S_{\boldsymbol{\mu},\boldsymbol{\Sigma}}
(\boldsymbol{\nu}_\lambda,\boldsymbol{\Sigma})$.
Assume that $\boldsymbol{\nu}_\lambda\neq\boldsymbol{0}$ for all
$\lambda\in[0,1]$. Then $S(\lambda)$ is nondecreasing on $[0,1]$. If $\boldsymbol{\mu}$ and $\boldsymbol{\nu}$ are non-collinear, then
$S(\lambda)$ is strictly increasing on $[0,1]$, with $S'(1)=0$.
$S(\lambda)$ is concave around 1. Moreover, $S(\lambda)$ is concave on $[0,1]$
if and only if $A+3B\ge0$.
\end{theorem}

By Theorem~\ref{thm:mean-improvement}, $S(\lambda)$ is concave on all of $[0,1]$
iff $A+3B\ge0$; since $A\ge0$, this holds whenever $B=\langle\boldsymbol\nu,
\boldsymbol\mu-\boldsymbol\nu\rangle\ge0$, i.e., $\cos\theta'>0$, meaning the
future view stays reasonably close to the historical mean. Figure~\ref{fig:sharpe-bad}
confirms this on Dataset~1 (\S~\ref{appendix:data}), where the condition holds and
the realized Sharpe ratio is concave in $[0,1]$. When $A+3B<0$, concavity
is guaranteed only near $\lambda=1$; Figure~\ref{fig:sharpe-good} shows the resulting
inflection point where the curve shifts from convex to concave.

\begin{figure}[h]
    \centering
    \begin{subfigure}[t]{0.49\textwidth}
        \centering
        \includegraphics[width=\linewidth]{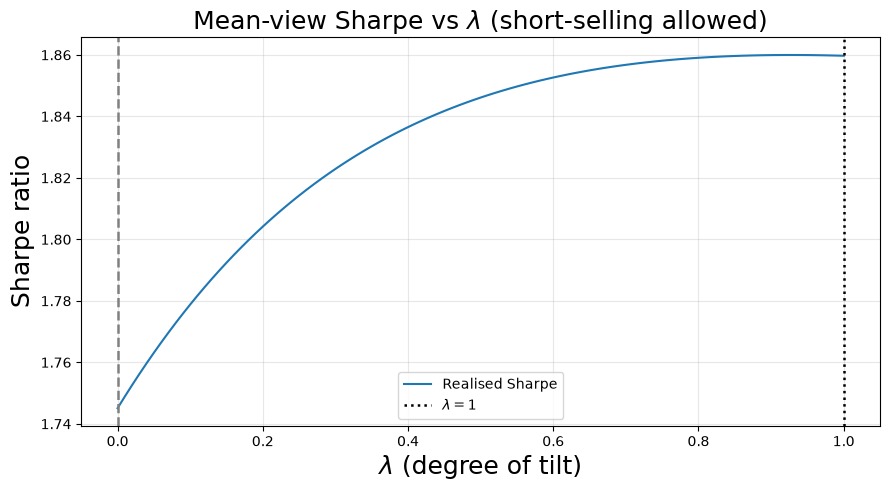}
        \caption{Sharpe vs $\lambda$ on Dataset 1}
        \label{fig:sharpe-bad}
    \end{subfigure}
    \hfill
    \begin{subfigure}[t]{0.49\textwidth}
        \centering
        \includegraphics[width=\linewidth]{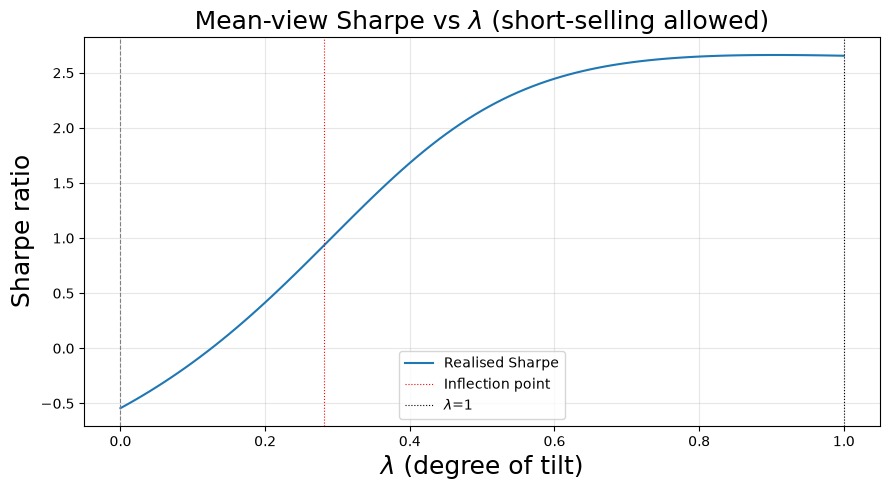}
        \caption{Sharpe vs $\lambda$ on Dataset 2}
        \label{fig:sharpe-good}
    \end{subfigure}
    \caption{Concavity of Sharpe as a function of $\lambda$}
\end{figure}

We prove Proposition ~\ref{prop:sharpe-angle} and an analogous statement for covariance views as well as Theorem~\ref{thm:mean-improvement} in \S~\ref{sec:thrandp}.

\section{Sensitivity of View Incorporation in the Markowitz Framework}\label{expmare}

In this section, we explain how expert views on the mean of returns can be theoretically incorporated in the Markowitz framework. We also state the perturbation bound Theorem \ref{prop:srd}, subsequently proved in \S~\ref{appendix:perturbation}. Throughout this section, we assume the existence of a historical measure $P$ on returnss (see \S~\ref{appendix:empiricalest}).

In the mean-view setup, an expert specifies the mean $\boldsymbol{\mu}^{\text{view}}$ of the future distribution. We incorporate the expert view as in \cite{black1992global}. For two probability distributions $Q$ and $P$ with densities $q$ and $p$ on the same support, the Kullback-Leibler (KL) divergence of $Q$ from $P$ is defined as
\[
D_{KL}(Q \,\|\, P) \;=\; \mathbb{E}_Q\!\left[\log \frac{q(\mathbf{r})}{p(\mathbf{r})}\right] \;=\; \int q(\mathbf{r}) \log \frac{q(\mathbf{r})}{p(\mathbf{r})}\, d\mathbf{r}.
\]
Incorporation is performed by changing the historical measure $P$ to $Q$, where $Q$ is the solution to the constrained problem $Q^\star \;=\; \mbox{argmin}_{Q \,:\, \mathbb{E}_Q[\mathbf{r}]=\boldsymbol{\mu}_\lambda^{\text{view}}} D_{KL}(Q \,\|\, P)$. If we assume that $\mathbf{r}_t \sim \mathcal{N}(\boldsymbol{\mu}, \boldsymbol{\Sigma})$ under $P$, then it is well-known that $\mathbf{r}_t \sim \mathcal{N}(\boldsymbol{\mu}^{\text{view}}, \boldsymbol{\Sigma})$ under $Q$. Theorem~\ref{thm:mean-improvement} applies to this particular scenario. However, under non-Gaussian returns, a mean view may also perturb the covariance matrix. We show that this perturbation costs only a \emph{second-order} loss in Sharpe ratio, assuming that the perturbation is small enough relative to the sizes of the true mean and variance. In this sense, our earlier results on directionally correct mean views are robust.

Recall $\operatorname{S}_{\boldsymbol{\mu},\boldsymbol{\Sigma}}$ and $\operatorname{S}^{\star}_{\boldsymbol{\mu},\boldsymbol{\Sigma}}$ from \eqref{eq:misspecified-sharpe}. For a symmetric matrix $A$, let $\|A\|_{op} = \sup_{\mathbf{x} \in \mathbb R^n \setminus \{0\}} \frac{\|A\mathbf{x}\|}{\|\mathbf{x}\|}.$
\begin{theorem}\label{prop:srd}
    Let $\boldsymbol{\Sigma}, \boldsymbol{\Omega} \succ 0$, and $\boldsymbol\mu, \boldsymbol\nu$ be fixed. For some $m>0$, suppose that $\|\boldsymbol{\mu}\|\geq m, \boldsymbol{\Sigma} - mI \succ 0$ hold. Furthermore, assume that $\|\boldsymbol{\mu}-\boldsymbol{\nu}\|< \frac m2, \|\boldsymbol{\Omega} - \boldsymbol{\Sigma}\|_{op} < \frac m2$. Then,  $$  0 \leq \frac{\operatorname{S}^{\star}_{\boldsymbol{\mu},\boldsymbol{\Sigma}} - \operatorname{S}_{\boldsymbol{\mu},\boldsymbol{\Sigma}} (\boldsymbol{\nu}, \boldsymbol{\Omega})}{\operatorname{S}^{\star}_{\boldsymbol{\mu},\boldsymbol{\Sigma}}} <  \frac{2\|\boldsymbol{\Sigma}^{\frac 12}\|_{op}^4}{m^2} \left(\|\boldsymbol\mu-\boldsymbol\nu\|+ \|\boldsymbol \Omega - \boldsymbol\Sigma\|_{op}\right)^2.$$
\end{theorem}

\section{Expert Views in the End-to-End Framework}\label{sec:end-to-end}

The end-to-end approach to portfolio optimization infers asset allocations directly from
historical data, bypassing the intermediate step of estimating the return distribution
\citep{zhang2020deep, zhang2021universal}. We incorporate mean views into this framework
through an exponential change of measure. A neural network $f_\phi$ maps past market
information to portfolio weights: the weight selected on day $t$ is a function of the
$L$-day lookback window preceding the decision date.

We work with $n$ assets and $T+L$ days of return data $x_t = (x_{t,1},\dots,x_{t,n})^\top$,
$t = -L+1,\dots,T$. For each day $i=0,\dots,T-1$, the network outputs the next day's
portfolio weight $w_{i+1} = f_\phi(x_{i-L+1},\dots,x_i)$, yielding return
$R_i(\phi) = x_{i+1}^\top w_{i+1}$. Following \cite{zhang2020deep}, $\phi$ is trained by
gradient ascent directly on the Sharpe ratio of these portfolio returns, rather than an intermediate return-prediction loss. The only difference being that we modify the Sharpe ratio via likelihood ratios to incorporate directionally correct views.

Let
$\boldsymbol{\mu}^{\text{hist}} = \mathbb{E}_P[\mathbf r]$ be the mean under the empirical
distribution $P$ (see \S~\ref{appendix:empiricalest} for its construction) and $\boldsymbol{\mu}^{\text{view}}$ the
mean implied by the expert view; for $\lambda\in[0,1]$, set
$\boldsymbol{\mu}^{\text{view}}_\lambda = (1-\lambda)\boldsymbol{\mu}^{\text{hist}} +
\lambda\boldsymbol{\mu}^{\text{view}}$. As in \S~\ref{expmare}, we reweight via the
likelihood ratio
\[
\ell_\lambda(\mathbf r) = \frac{\exp(\boldsymbol{\eta}_\lambda^\top \mathbf r)}
{\mathbb{E}_P[\exp(\boldsymbol{\eta}_\lambda^\top \mathbf r)]}, \qquad
\mathbb{E}_P[\ell_\lambda(\mathbf r)\mathbf r] = \boldsymbol{\mu}^{\text{view}}_\lambda.
\]
Since $P$ is uniform on $\mathbf r_1,\dots,\mathbf r_T$, this gives the explicit weights
$q_{t,\lambda} = \exp(\boldsymbol{\eta}_\lambda^\top \mathbf r_t) /
\sum_{s=1}^T \exp(\boldsymbol{\eta}_\lambda^\top \mathbf r_s)$, with
$\sum_{t=1}^T q_{t,\lambda}\mathbf r_t = \boldsymbol{\mu}^{\text{view}}_\lambda$.

With $m_\lambda(\phi) = \mathbb{E}_P[\ell_\lambda(\mathbf r_t) R_t(\phi)]$ and
$v_\lambda(\phi) = \mathbb{E}_P[\ell_\lambda(\mathbf r_t)(R_t(\phi)-m_\lambda(\phi))^2]$, we
update $\phi$ by gradient ascent on the view-adjusted Sharpe ratio
$S_{\phi,\lambda} = m_\lambda(\phi)/\sqrt{v_\lambda(\phi)}$.

\textbf{Experiments.} We test this framework against the Markowitz model of
\S~\ref{expmare} on daily NSE India data for the 20 largest firms by market
capitalization (2006--2025; see \S~\ref{appendix:data}). Both models are given a view
$\boldsymbol{\mu}_\lambda$, $\lambda \in \{0, 0.1, \dots, 1\}$, interpolating between the
historical mean (2006--2023) and the realized mean over an evaluation period
(2024--2025), and are evaluated on that period; $\lambda=1$ represents a perfectly
accurate view unavailable at construction time. Given constraints on short-selling in
Indian markets, we restrict both experiments to long-only portfolios, giving a common
constraint under which to compare the two approaches.

\textbf{Models.} For Markowitz, the tilted probabilities $q_t(\lambda)$ that impose
$\boldsymbol{\mu}_\lambda$ also determine a covariance estimate\footnote{We estimate $\boldsymbol{\Sigma}_\lambda$ via the tilted probabilities
$q_t(\lambda)$ rather than by bootstrapping. We found this to be less noisy in practice.}
$\boldsymbol{\Sigma}_\lambda = \sum_{t=1}^T q_t(\lambda)(\mathbf r_t -
\boldsymbol{\mu}_\lambda)(\mathbf r_t - \boldsymbol{\mu}_\lambda)^\top$; using
$(\boldsymbol{\mu}_\lambda, \boldsymbol{\Sigma}_\lambda)$, we solve for the long-only
maximum-Sharpe weights $\mathbf w_\lambda \in \arg\max_{\mathbf w \geq 0,\, \mathbf
1^\top \mathbf w = 1} \mathbf w^\top \boldsymbol{\mu}_\lambda / \sqrt{\mathbf w^\top
\boldsymbol{\Sigma}_\lambda \mathbf w}$. For the end-to-end model, the network takes a
lookback window of past returns and outputs a long-only, fully invested portfolio
through a final softmax layer; the reported architecture is a single-layer LSTM with 16
hidden units, chosen after exploratory experiments with fully connected networks, CNNs,
and larger or deeper LSTMs, none of which improved on it, consistent with the findings
of \cite{zhang2020deep}. For each $\lambda$, we train 100 random initializations under
an identical protocol: 300 full-batch updates, with a 50-update warm-up followed by
cosine annealing over the remaining 250.

\textbf{Results.} Figure~\ref{fig:performance-comparison} compares the Markowitz
portfolio with the mean performance of the 100 LSTM runs, reported after a one-way
transaction cost of 20 bps on the value traded.\footnote{This is a conservative
estimate: delivery-based equity transactions in India carry a securities transaction tax
of 10 bps on both purchases and sales, before brokerage and other costs.}\footnote{If
$\mathbf w_t$ is the target allocation and $\mathbf w_{t-1}$ the allocation immediately
before rebalancing, the cost deducted from the daily return is $c\sum_{i=1}^n
\left|w_{t,i}-w_{t-1,i}\right|.$} Both improve sharply as $\lambda$ increases, with much
of the gain arriving already at small $\lambda$ and tapering off as the view approaches
the future mean; the two Sharpe curves are strikingly similar throughout, and dispersion
across LSTM seeds falls sharply for more informative views. After costs, the LSTM
performs slightly worse than Markowitz, particularly at intermediate $\lambda$,
consistent with its greater turnover. Exact Sharpe ratios and drawdowns for each
$\lambda$, along with portfolio allocations at $\lambda = 0, 0.5, 1$, are given in
\S~\ref{appendix:experimentresults}.

\begin{figure}[h]
    \centering
    \begin{subfigure}[t]{0.49\textwidth}
        \centering
        \includegraphics[width=\linewidth]{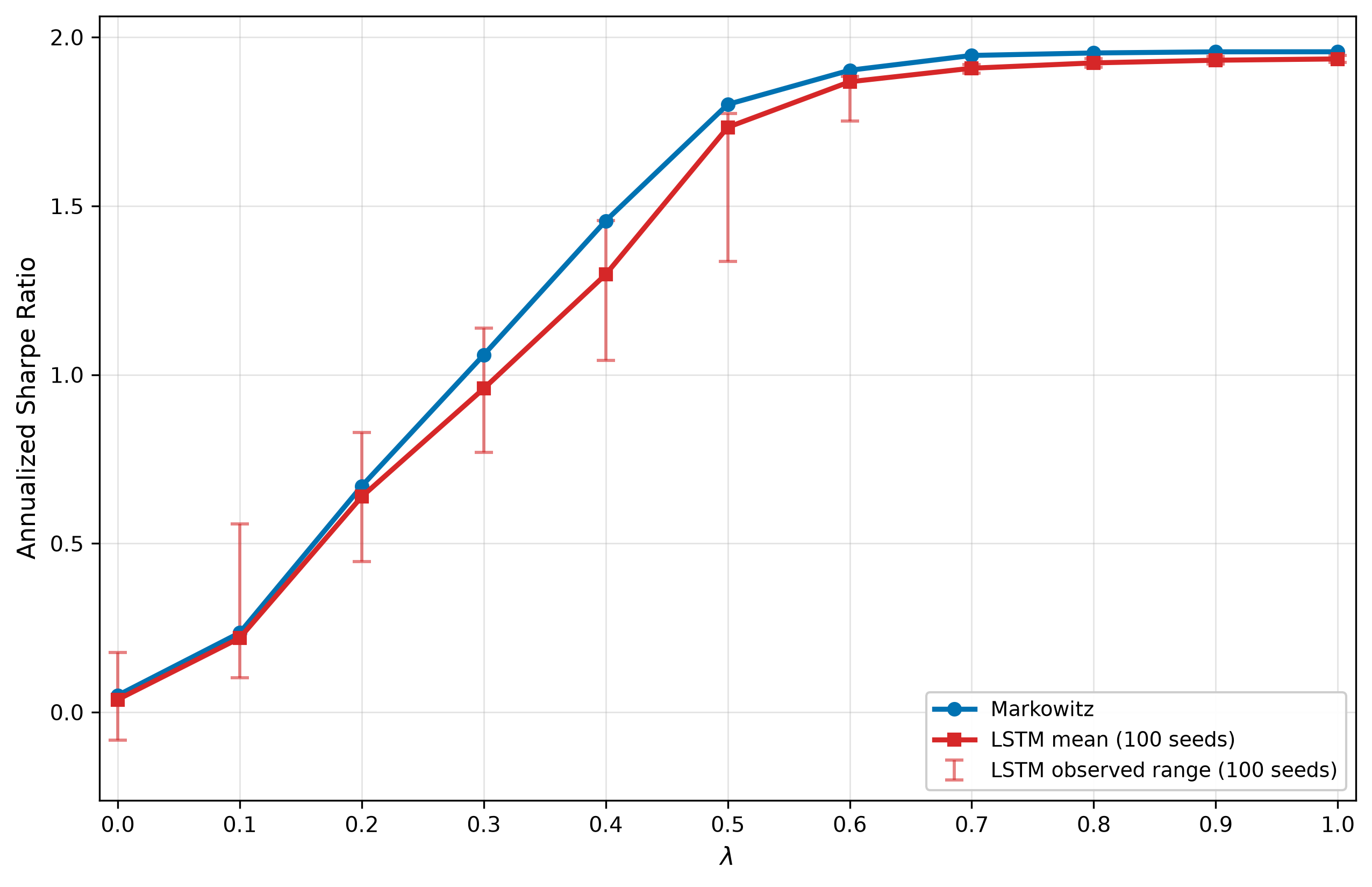}
        \caption{Annualized Sharpe ratio.}
        \label{fig:sharpe-comparison}
    \end{subfigure}
    \hfill
    \begin{subfigure}[t]{0.49\textwidth}
        \centering
        \includegraphics[width=\linewidth]{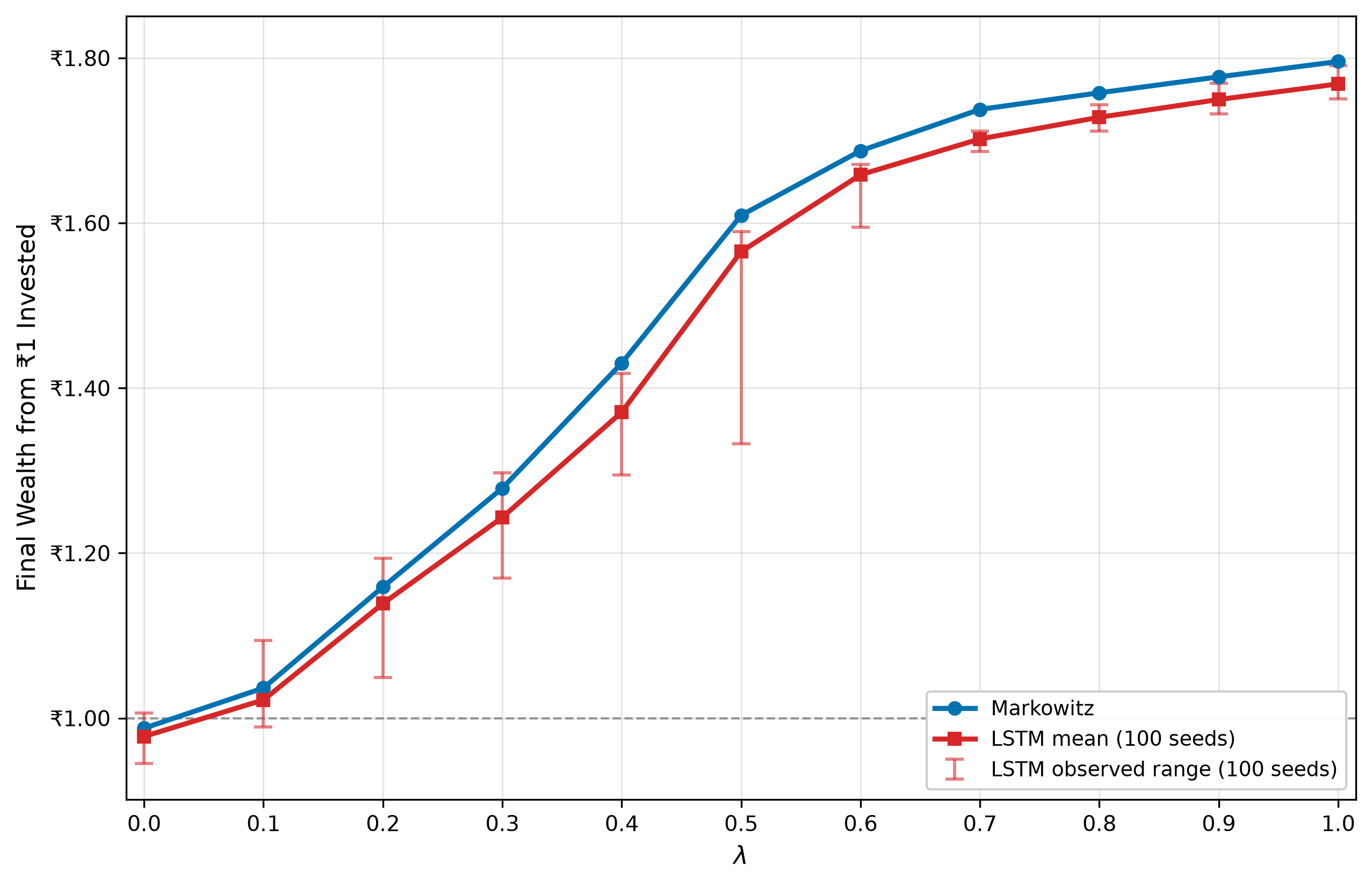}
        \caption{Final value of INR~1 investment.}
        \label{fig:wealth-comparison}
    \end{subfigure}
    \caption{Performance of Markowitz vs End-To-End across $\lambda$.}
    \label{fig:performance-comparison}
\end{figure}

\section{Conclusion}\label{conc}

In this article, we demonstrated that incorporating directionally correct expert views can improve portfolio design. We also showed that the Sharpe ratio is stable to minor errors in these views, heuristically demonstrating why complex ML models learn allocations similar to those predicted by Markowitz theory and, in turn, achieve comparable Sharpe ratios in our experiments.

Having established the value of directionally correct views, a natural next step is producing such views in practice, since our experiments constructed them using future information unavailable at decision time. We leave this to future work.


\bibliographystyle{unsrt}
\bibliography{reffornw}

\section{Theoretical results and proofs}\label{sec:thrandp}

In this section, we state and prove the theoretical results in this article, including any that were mentioned in the main body.

\subsection{Proofs in Section~\ref{sec:srunderev}}\label{sec:pfsinsec2}

In this section, we prove the results in Section~\ref{sec:srunderev}. Recall the definition of the Sharpe ratio \eqref{eq:SR}, and that $(\mu,\Sigma)$ are the true mean and variance of future returns, while $(\nu,\Omega)$ are the mean and variance estimated by any method.

\subsubsection{Proof of Proposition~\ref{prop:sharpe-angle} and Theorem~\ref{thm:mean-improvement}}

We begin with the proof of Proposition~\ref{prop:sharpe-angle}

\begin{proof}[Proof of Proposition~\ref{prop:sharpe-angle}]
Since $\boldsymbol{\Sigma}\succ0$, the bilinear form
 $    \langle \mathbf x,\mathbf y\rangle_{\boldsymbol{\Sigma}}
    :=
    \mathbf x^\top\boldsymbol{\Sigma}\mathbf y$ defines an inner product on $\mathbb R^n$. By the Cauchy-Schwarz inequality applied to this inner product, for any $\mathbf{x},\mathbf{y} \in \mathbb R^n$,
\[
    |\mathbf x^\top\mathbf y|
    \leq
    \sqrt{\mathbf x^\top\boldsymbol{\Sigma}^{-1}\mathbf x}\,
    \sqrt{\mathbf y^\top\boldsymbol{\Sigma}\mathbf y}.
\]
Taking $ \mkern7mu
    \mathbf x=\boldsymbol{\mu}
    $ and $
    \mathbf y=\boldsymbol{\Omega}^{-1}\boldsymbol{\nu} \mkern7mu$ we have
\[
\left|
\frac{
    \boldsymbol{\mu}^{\top}\boldsymbol{\Omega}^{-1}\boldsymbol{\nu}
}{
    \sqrt{
        \boldsymbol{\mu}^{\top}\boldsymbol{\Sigma}^{-1}\boldsymbol{\mu}
    }
    \sqrt{
        \boldsymbol{\nu}^{\top}
        \boldsymbol{\Omega}^{-1}\boldsymbol{\Sigma}
        \boldsymbol{\Omega}^{-1}\boldsymbol{\nu}
    }
}
\right|
\leq 1.
\]
Since $\boldsymbol{\mu},\boldsymbol{\nu}\neq0$ and
$\boldsymbol{\Sigma},\boldsymbol{\Omega}\succ0$, both terms in the
denominator are strictly positive. Thus there is a unique
$\theta_{\boldsymbol{\mu},\boldsymbol{\Sigma}}
(\boldsymbol{\nu},\boldsymbol{\Omega})\in[0,\pi]$
whose cosine is the quantity above. Equivalently, this is the angle
under $\langle\cdot,\cdot\rangle_{\boldsymbol{\Sigma}}$ between
$\boldsymbol{\Sigma}^{-1}\boldsymbol{\mu}$ and
$\boldsymbol{\Omega}^{-1}\boldsymbol{\nu}$.
\\
Now, through minor calculations, it is easily verified that $S_{\boldsymbol{\mu},\boldsymbol{\Sigma}}
(\boldsymbol{\nu},\boldsymbol{\Omega}) = S^\star_{\boldsymbol{\mu},\boldsymbol{\Sigma}}
\cos
\theta_{\boldsymbol{\mu},\boldsymbol{\Sigma}}
(\boldsymbol{\nu},\boldsymbol{\Omega})$.
Therefore $  S_{\boldsymbol{\mu},\boldsymbol{\Sigma}}
(\boldsymbol{\nu},\boldsymbol{\Omega})
\leq
S^\star_{\boldsymbol{\mu},\boldsymbol{\Sigma}}$,
with equality if and only if the angle is zero. By the equality
condition in Cauchy-Schwarz, this occurs  iff $\boldsymbol{\Omega}^{-1}\boldsymbol{\nu}
    =
    c\,\boldsymbol{\Sigma}^{-1}\boldsymbol{\mu}$
for some $c>0$.
\end{proof}

Next, we prove Theorem~\ref{thm:mean-improvement}.

\begin{proof}[Proof of Theorem~\ref{thm:mean-improvement}]
Using $\boldsymbol{\nu}_\lambda
=\boldsymbol{\nu}+\lambda(\boldsymbol{\mu}-\boldsymbol{\nu})$, we can write
\[
S(\lambda)
=
\frac{A+B+\lambda(B+C)}
{\bigl(A+2\lambda B+\lambda^2C\bigr)^{1/2}} \implies 
S'(\lambda)
=
\frac{(1-\lambda)(AC-B^2)}
{\bigl(A+2\lambda B+\lambda^2C\bigr)^{3/2}}.
\]
By the Cauchy-Schwarz inequality, $B^2\le AC$ with equality if and only if $\boldsymbol{\nu}$ and $\boldsymbol{\mu}-\boldsymbol{\nu}$ are collinear, which is equivalent to $\boldsymbol{\mu}$ and $\boldsymbol{\nu}$ being collinear. Hence $S'(\lambda)\ge0$ on $[0,1]$ and
$S'(1)=0$.

For concavity, differentiating once more yields
\begin{equation}\label{sprpr}
S''(\lambda)
=
-\,\frac{(AC-B^2)H(\lambda)}
{\bigl(A+2\lambda B+\lambda^2C\bigr)^{5/2}},
\end{equation}
where
$ 
H(\lambda)
=
A+3B+(3C-B)\lambda-2C\lambda^2.
$ 
Since $C\ge0$, $H$ is concave in $\lambda$. At the endpoints, $ H(0)=A+3B$, while $H(1)
=\langle\boldsymbol{\mu},\boldsymbol{\mu}\rangle_{\boldsymbol{\Sigma}^{-1}}>0$. Therefore, $H(\lambda)\ge0$ for every $\lambda\in[0,1]$ if and only if
$A+3B\ge0$. 

In the non-collinear case, $AC-B^2>0$, hence $S''(\lambda)<0$ throughout $[0,1]$ by \eqref{sprpr}. 
\end{proof}

\subsubsection{Covariance Views}\label{appendix:covarianceviews}

In this section, we state and prove an analogue of Theorem~\ref{thm:mean-improvement} for covariance views (Theorem~\ref{thm:cov-improvement}).

For a fixed symmetric matrix $A$ and $k \in \mathbb{R}$, consider a view of the form $\mbox{tr}(A \boldsymbol{\Omega}) =k$, where $\Omega$ is the covariance matrix of future returns. Interestingly, for Gaussian returns, incorporating this view corresponds to a linear shift of the inverse matrix.

\begin{proposition}
\label{prop:precision-shift}
Let $\boldsymbol{\Sigma}\succ0$ and consider
\[
\boldsymbol{\Omega}^\star = \arg\min_{\boldsymbol{\Omega}\succ0} \; D_{\mathrm{KL}}\big(\mathcal N(\boldsymbol{\mu},\boldsymbol{\Omega})\,\|\,\mathcal N(\boldsymbol{\mu},\boldsymbol{\Sigma})\big)
\quad \text{subject to} \quad \operatorname{tr}(A\boldsymbol{\Omega}) = k,
\]
for some fixed symmetric $A$ and $k \in \mathbb{R}$. Then $\boldsymbol{\Omega}^\star = (\boldsymbol{\Sigma}^{-1} + \lambda A)^{-1},
$ 
for some scalar $\lambda\in\mathbb{R}$ determined by the constraint. 
\end{proposition}

\begin{proof}
An explicit formula for the KL divergence between multivariate Gaussians can be found in \cite[(21.171)]{murphy2012machine}. Using this formula,
\[
D_{\mathrm{KL}}\big(\mathcal N(\boldsymbol{\mu},\boldsymbol{\Omega})\,\|\,\mathcal N(\boldsymbol{\mu},\boldsymbol{\Sigma})\big)
= \tfrac12\Big(\operatorname{tr}(\boldsymbol{\Sigma}^{-1}\boldsymbol{\Omega}) - \log\det(\boldsymbol{\Sigma}^{-1}\boldsymbol{\Omega}) - n\Big).
\]
Thus, it is sufficient to minimize $f(\boldsymbol{\Omega}) = \operatorname{tr}(\boldsymbol{\Sigma}^{-1}\boldsymbol{\Omega}) - \log\det\boldsymbol{\Omega}$, which is strictly convex on $\{\boldsymbol{\Omega}\succ0\}$. We form the Lagrangian
\[
\mathcal{L}(\boldsymbol{\Omega},\lambda) = \operatorname{tr}(\boldsymbol{\Sigma}^{-1}\boldsymbol{\Omega}) - \log\det\boldsymbol{\Omega} + \lambda\big(\operatorname{tr}(A\boldsymbol{\Omega}) - k\big).
\]
Using the matrix identities $\partial_{\boldsymbol{\Omega}}\operatorname{tr}(\boldsymbol{\Sigma}^{-1}\boldsymbol{\Omega}) = \boldsymbol{\Sigma}^{-1}$ and $\partial_{\boldsymbol{\Omega}}\log\det\boldsymbol{\Omega} = \boldsymbol{\Omega}^{-1}$, the stationarity condition $\partial_{\boldsymbol{\Omega}}\mathcal{L}=0$ gives
\[
\boldsymbol{\Sigma}^{-1} - \boldsymbol{\Omega}^{-1} + \lambda A = 0 \quad\Longleftrightarrow\quad \boldsymbol{\Omega}^{-1} = \boldsymbol{\Sigma}^{-1} + \lambda A.
\]
Since $f$ is strictly convex, this stationary point is the unique minimizer whenever the constraint is feasible, completing the proof.
\end{proof}

Now, let $P=\mathcal{N}(\boldsymbol{\mu}, \boldsymbol{\Sigma})$ be our estimate and $Q=\mathcal{N}(\boldsymbol{\mu}, \boldsymbol{\Omega})$ be the true future distribution of asset returns. As for mean views, we call a covariance view as directionally correct if the forecast is of the form $\boldsymbol{\Omega}(t)=\left(\boldsymbol{\Sigma}^{-1}+t\Delta\right)^{-1}$ for some $t \in [0,1]$, where $\Delta:=\boldsymbol{\Omega}^{-1}-\boldsymbol{\Sigma}^{-1}$.

By Proposition~\ref{prop:sharpe-angle} the optimal allocation has weights $\mathbf{w}(t)\propto \boldsymbol{\Omega}(t)^{-1}\boldsymbol{\mu}$. Let $S(t) =\operatorname{S}_{\boldsymbol{\mu},\boldsymbol{\Sigma}} (\boldsymbol{\mu}, \boldsymbol{\Omega}(t))$. The following analogue of Theorem~\ref{thm:mean-improvement} exists in this setting.

\begin{theorem}[Monotonicity and concavity]
\label{thm:cov-improvement}
Let the Sharpe realized under the true distribution
$\mathcal N(\boldsymbol\mu,\boldsymbol\Omega)$, of the tangency portfolio
$\mathbf w(t)\propto\boldsymbol\Omega(t)^{-1}\boldsymbol\mu$, be denoted by
\[
    S(t) := S_{\boldsymbol\mu,\boldsymbol\Omega}\big(\boldsymbol\mu,\boldsymbol\Omega(t)\big)
\]
Then $S(t)$ is nondecreasing on $[0,1]$. Further, if $\mathbf p$ and $\mathbf q$ are not collinear, then $S(t)$ is strictly
increasing on $[0,1]$. Moreover, $S(t)$ is concave on
$[0,1]$ if and only if $A+3B\ge0$; in particular $B\ge0$ is sufficient
for concavity.
\end{theorem}
Prior to proving the theorem, we introduce the notation
$ 
    A:=\langle\mathbf p,\mathbf p\rangle_{\boldsymbol\Omega},
    \mkern7mu
    B:=\langle\mathbf p,\mathbf q-\mathbf p\rangle_{\boldsymbol\Omega},
    \mkern7mu
    C:=\langle\mathbf q-\mathbf p,\mathbf q-\mathbf p\rangle_{\boldsymbol\Omega},
$
where $ \mkern7mu
    \mathbf p:=\boldsymbol\Sigma^{-1}\boldsymbol\mu,
    \mkern7mu
    \mathbf q:=\boldsymbol\Omega^{-1}\boldsymbol\mu,
    \mkern7mu
    \langle\mathbf x,\mathbf y\rangle_{\boldsymbol\Omega}
    :=\mathbf x^\top\boldsymbol\Omega\mathbf y.
 $ 

\begin{proof}
Since $\Delta\boldsymbol\mu=\boldsymbol\Omega^{-1}\boldsymbol\mu-
\boldsymbol\Sigma^{-1}\boldsymbol\mu=\mathbf q-\mathbf p$, the defining
relation for $\boldsymbol\Omega(t)^{-1}$ gives
\[
    \boldsymbol\Omega(t)^{-1}\boldsymbol\mu
    =
    \boldsymbol\Sigma^{-1}\boldsymbol\mu+t\Delta\boldsymbol\mu
    =
    \mathbf p+t(\mathbf q-\mathbf p)
    =:
    \mathbf p_t,
\]
This is exactly the structure
of $\boldsymbol\nu_\lambda$ in Theorem~\ref{thm:mean-improvement}. Assuming
$\boldsymbol\Omega(t)\succ0$ and $\boldsymbol\mu\neq\mathbf 0$, we have
$\mathbf p_t=\boldsymbol\Omega(t)^{-1}\boldsymbol\mu\neq\mathbf 0$ for every
$t\in[0,1]$. By \eqref{eq:misspecified-sharpe},
\[
    S(t)
    =
    \frac{\boldsymbol\mu^\top\mathbf p_t}
    {\sqrt{\mathbf p_t^\top\boldsymbol\Omega\mathbf p_t}}.
\]
Using $\boldsymbol\Omega\mathbf q=\boldsymbol\mu$, the numerator satisfies
$\boldsymbol\mu^\top\mathbf p_t=\mathbf q^\top\boldsymbol\Omega\mathbf p_t
=\langle\mathbf q,\mathbf p_t\rangle_{\boldsymbol\Omega}$,
\[
    S(t)
    =
    \frac{\langle\mathbf q,\mathbf p_t\rangle_{\boldsymbol\Omega}}
    {\sqrt{\langle\mathbf p_t,\mathbf p_t\rangle_{\boldsymbol\Omega}}}
    =
    \frac{A+B+t(B+C)}
    {\bigl(A+2tB+t^2C\bigr)^{1/2}},
\]
where the second equality follows from expanding
$\mathbf p_t=\mathbf p+t(\mathbf q-\mathbf p)$ in the inner product
$\langle\cdot,\cdot\rangle_{\boldsymbol\Omega}$ and using the definitions
of $A,B,C$. Differentiating,
\[
    S'(t)
    =
    \frac{(1-t)(AC-B^2)}
    {\bigl(A+2tB+t^2C\bigr)^{3/2}}.
\]

This is exactly the expression for $S'(\lambda)$ obtained in the proof of Theorem~\ref{thm:mean-improvement}. Thus, the same analysis follows, giving us monotonicity of the Sharpe ratio as a function of $t$, and concavity iff $A+3B \ge 0$.
\end{proof}

\subsection{Proof of Theorem~\ref{prop:srd}}
\label{appendix:perturbation}

\begin{proof}[Proof of Theorem~\ref{prop:srd}]
Suppose that $\boldsymbol{\mu},\boldsymbol{\nu}, \boldsymbol{\Sigma},\boldsymbol{\Omega},m$ are as in the statement of the theorem. By Proposition~\ref{prop:sharpe-angle} we have \begin{equation}\label{eq1}
\frac{S^*_{\boldsymbol{\mu}, \boldsymbol{\Sigma}} - S_{\boldsymbol{\mu}, \boldsymbol{\Sigma}}(\boldsymbol{\nu},\boldsymbol{\Omega})}{S^*_{\boldsymbol{\mu}, \boldsymbol{\Sigma}}} = 1- \cos \theta_{\boldsymbol{\mu},\boldsymbol{\Sigma}}(\boldsymbol{\nu}, \boldsymbol{\Omega}),
\end{equation}
where\begin{equation}\label{eq2}
\cos
\theta_{\boldsymbol{\mu},\boldsymbol{\Sigma}}
(\boldsymbol{\nu},\boldsymbol{\Omega})
=
\frac{
    \boldsymbol{\mu}^{\top}\boldsymbol{\Omega}^{-1}\boldsymbol{\nu}
}{
    \sqrt{
        \boldsymbol{\mu}^{\top}
        \boldsymbol{\Sigma}^{-1}\boldsymbol{\mu}
    }
    \sqrt{
        \boldsymbol{\nu}^{\top}
        \boldsymbol{\Omega}^{-1}
        \boldsymbol{\Sigma}
        \boldsymbol{\Omega}^{-1}\boldsymbol{\nu}
    }
}.
\end{equation}
The left side of the inequality in the theorem follows directly from Proposition~\ref{prop:sharpe-angle}. Thus, we must only prove the right hand side.

To do this, we start by noticing that if we take $\mathbf{x} = \boldsymbol{\Sigma}^{-\frac 12} \boldsymbol{\mu}$, $\mathbf{y} = \boldsymbol{\Sigma}^{-\frac 12} \boldsymbol{\nu}$ and $A = \boldsymbol{\Sigma}^{-\frac 12}\boldsymbol{\Omega}\boldsymbol{\Sigma}^{-\frac 12}$, then 
\begin{equation}\label{eq3}
1-\frac{
    \boldsymbol{\mu}^{\top}\boldsymbol{\Omega}^{-1}\boldsymbol{\nu}
}{
    \sqrt{
        \boldsymbol{\mu}^{\top}
        \boldsymbol{\Sigma}^{-1}\boldsymbol{\mu}
    }
    \sqrt{
        \boldsymbol{\nu}^{\top}
        \boldsymbol{\Omega}^{-1}
        \boldsymbol{\Sigma}
        \boldsymbol{\Omega}^{-1}\boldsymbol{\nu}
    }
} = 1-\frac{\mathbf{x}^T A^{-1} \mathbf{y}}{\|\mathbf{x}\| \|A^{-1}\mathbf{y}\|} = \frac 12\left\|\frac{\mathbf{x}}{\|\mathbf{x}\|} - \frac{A^{-1}\mathbf{y}}{\|A^{-1}\mathbf{y}\|}\right\|^2.
\end{equation}
Now, for any two non-zero vectors $\mathbf{u},\mathbf{v}$ we have the inequality $$
\left\|\frac{\mathbf{u}}{\|\mathbf{u}\|} - \frac{\mathbf{v}}{\|\mathbf{v}\|}\right\|^2 \leq \frac{4\|\mathbf{u}-\mathbf{v}\|^2}{\|\mathbf{u}\|^2}.
$$
Using this inequality with $\mathbf{u} = \mathbf{x}, \mathbf{v} = A^{-1}\mathbf{y}$ in \eqref{eq3} we have 
\begin{equation}\label{eq4}
1-\frac{
    \boldsymbol{\mu}^{\top}\boldsymbol{\Omega}^{-1}\boldsymbol{\nu}
}{
    \sqrt{
        \boldsymbol{\mu}^{\top}
        \boldsymbol{\Sigma}^{-1}\boldsymbol{\mu}
    }
    \sqrt{
        \boldsymbol{\nu}^{\top}
        \boldsymbol{\Omega}^{-1}
        \boldsymbol{\Sigma}
        \boldsymbol{\Omega}^{-1}\boldsymbol{\nu}
    }
} \leq 2 \frac{\|A^{-1}\mathbf{y} - \mathbf{x}\|^2}{\|\mathbf{x}\|^2} = \frac{\|\boldsymbol{\Sigma}^{\frac 12}\boldsymbol{\Omega}^{-1}\boldsymbol{\nu} - \boldsymbol{\Sigma}^{-\frac 12}\boldsymbol{\mu}\|^2}{\boldsymbol{\mu}^T \boldsymbol{\Sigma}^{-1} \boldsymbol{\mu}}.
\end{equation}
We shall first upper bound the norm in the numerator, by the triangle inequality and the definition of the operator norm $\|\cdot\|_{op}$, which we recall from above the statement of this theorem.
\begin{align}
\|\boldsymbol{\Sigma}^{\frac 12}\boldsymbol{\Omega}^{-1}\boldsymbol{\nu} - \boldsymbol{\Sigma}^{-\frac 12}\boldsymbol{\mu}\| \leq & \|\boldsymbol{\Sigma}^{\frac 12}\boldsymbol{\Omega}^{-1}\boldsymbol{\nu} - \boldsymbol{\Sigma}^{\frac 12}\boldsymbol{\Omega}^{-1}\boldsymbol{\mu}\| + \| \boldsymbol{\Sigma}^{\frac 12}\boldsymbol{\Omega}^{-1}\boldsymbol{\mu}-\boldsymbol{\Sigma}^{-\frac 12}\boldsymbol{\mu}\| \nonumber \\
\leq &  \|\boldsymbol{\Sigma}^{\frac 12}\|_{op}\|\boldsymbol{\Omega}^{-1}\|_{op}\|\boldsymbol{\nu} -\boldsymbol{\mu}\| +  \|\boldsymbol{\Sigma}^{\frac 12}\|_{op}\|\boldsymbol{\Omega}^{-1} -\boldsymbol{\Sigma}^{-1} \|_{op} \|\boldsymbol{\mu}\| \label{eq5}
\end{align}
Now, we apply the assumptions of the theorem to continue this bound. First, since $\boldsymbol{\Sigma} -mI\succ 0$, we have $\|\boldsymbol{\Sigma}^{-1}\|_{op} \leq \frac 1m$. Then, since $\|\boldsymbol{\Sigma} - \boldsymbol{\Omega}\|_{op}< \frac m2$, we obtain $\frac m2I - (\boldsymbol{\Sigma}-\boldsymbol{\Omega}) \succ 0$. Adding the two matrix inequalities, we obtain that $\boldsymbol{\Omega} - \frac m2I \succ 0$, which implies that $\|\boldsymbol{\Omega}^{-1}\|_{op} \leq \frac 2m$. Next, observe that $$
\|\boldsymbol{\Omega}^{-1} -\boldsymbol{\Sigma}^{-1} \|_{op} = \|\boldsymbol{\Sigma}^{-1}(\boldsymbol{\Omega} -\boldsymbol{\Sigma})\boldsymbol{\Omega}^{-1} \|_{op} \leq \|\boldsymbol{\Sigma}^{-1}\|_{op}\|\boldsymbol{\Omega} -\boldsymbol{\Sigma}\|_{op}\|\boldsymbol{\Omega}^{-1}\|_{op} \leq \frac 2{m^2}\|\boldsymbol{\Omega} -\boldsymbol{\Sigma}\|_{op}. 
$$
Putting these inequalities into \eqref{eq5} we have $$
\|\boldsymbol{\Sigma}^{\frac 12}\boldsymbol{\Omega}^{-1}\boldsymbol{\nu} - \boldsymbol{\Sigma}^{-\frac 12}\boldsymbol{\mu}\| \leq \frac {2}{m}\|\boldsymbol{\Sigma}^{\frac 12}\|_{op}\|\boldsymbol\mu-  \boldsymbol\nu\| + \frac 2{m^2}\|\boldsymbol{\Sigma}^{\frac 12}\|_{op} \|\boldsymbol\mu\|\|\boldsymbol{\Omega} -\boldsymbol{\Sigma}\|_{op}.
$$
We plug the above into \eqref{eq4}, and note that in \eqref{eq4}, the denominator is lower bounded: $$\boldsymbol{\mu}^T\boldsymbol{\Sigma}^{-1} \boldsymbol{\mu} = \|\boldsymbol\Sigma^{-\frac 12} \boldsymbol\mu\|^2 \geq \frac{\|\boldsymbol{\mu}\|^2}{\|\boldsymbol{\Sigma}^{\frac 12}\|_{op}^2}.$$ Thus, using the fact that $\|\boldsymbol{\mu}\|\geq m$,
$$
1-\frac{
    \boldsymbol{\mu}^{\top}\boldsymbol{\Omega}^{-1}\boldsymbol{\nu}
}{
    \sqrt{
        \boldsymbol{\mu}^{\top}
        \boldsymbol{\Sigma}^{-1}\boldsymbol{\mu}
    }
    \sqrt{
        \boldsymbol{\nu}^{\top}
        \boldsymbol{\Omega}^{-1}
        \boldsymbol{\Sigma}
        \boldsymbol{\Omega}^{-1}\boldsymbol{\nu}
    }
} \leq \frac{2\|\boldsymbol{\Sigma}^{\frac 12}\|_{op}^4}{m^2} \left(\|\boldsymbol\mu-\boldsymbol\nu\|+ \|\boldsymbol \Omega - \boldsymbol\Sigma\|_{op}\right)^2.
$$
This instantly concludes the proof in light of \eqref{eq1} and \eqref{eq2}.
\end{proof}

\section{Experiments}
\label{appendix:experiments}

This appendix collects supporting material for the experiments in
\S~\ref{sec:end-to-end}. \S~\ref{appendix:empiricalest} details the construction of the
historical measure used to form our initial view on returns, \S~\ref{appendix:data}
describes the equity universe and data underlying both experiments, and
\S~\ref{appendix:experimentresults} reports the full performance statistics and
portfolio allocations summarized in Figure~\ref{fig:performance-comparison}.

\subsection{Historical Measure}\label{appendix:empiricalest}

Let $\mathbf r_i \in \mathbb [-1, \infty)^n$, $i = 1,\ldots,T$, denote the return vectors for $n$
assets observed over a historical period of $T$ days. For these observations, we take $P$ to be the historical measure on $\mathbf
r_1,\ldots,\mathbf r_T$, defined for any subset $A \subset \mathbb R^n$ by
\[
P(A) = \frac{1}{T} \sum_{i=1}^T \mathbf{1}_{\mathbf r_i \in A},
\]
i.e., a uniform prior over the historical returns. If $\mathbf r$ denotes the random
return vector under $P$, then $\mathbb E_P[\mathbf r] = \frac{1}{T}\sum_{i=1}^T
\mathbf r_i$.

\subsection{Equity Universe and Data}
\label{appendix:data}

We restrict to firms with a closing price available on every trading date in the sample window from 1 January 2006 through 31 December 2025; since 1 January 2006 was a market holiday, the first observed trading date is 2 January 2006. This yields a common panel of 4,958 trading dates across the eligible firms. From this eligible pool, we rank firms by market capitalization as of 31 December 2023 and select the top 20.

The resulting universe consists of \texttt{RELIANCE}, \texttt{TCS}, \texttt{HDFCBANK}, \texttt{AXISBANK}, \texttt{INFY}, \texttt{HINDUNILVR}, \texttt{BHARTIARTL}, \texttt{ITC}, \texttt{SBIN}, \texttt{LT}, \texttt{HCLTECH}, \texttt{KOTAKBANK}, \texttt{ICICIBANK}, \texttt{ASIANPAINT}, \texttt{NTPC}, \texttt{ADANIENT}, \texttt{MARUTI}, \texttt{SUNPHARMA}, \texttt{ULTRACEMCO}, , and \texttt{TITAN}.

This also serves as the Dataset 2 in Figure 1. However, Dataset 1 contains \texttt{BHARTIARTL}, \texttt{KOTAKBANK}, \texttt{SBIN}, \texttt{MARUTI} trained on data from 1 January 2020 to 31 December, 2023. Dataset 1 aims to serve as an illustration of the theoretical concavity result when the necessary conditions hold.

\subsection{Portfolio Allocations and Performance Statistics}
\label{appendix:experimentresults}

Figure~\ref{fig:portfolio-composition} compares the Markowitz and mean LSTM allocations over 2024--2025 for three representative views, $\lambda=0$, $0.5$, and $1$. The Markowitz allocation is constant by construction, while the LSTM allocation can vary from day to day. With the historical view ($\lambda=0$), the LSTM makes substantial changes over time. As the view becomes more informative, these changes largely disappear and the LSTM converges toward an almost constant allocation that closely resembles the corresponding Markowitz portfolio. At $\lambda=1$, both portfolios are concentrated primarily in \texttt{BHARTIARTL}, \texttt{MARUTI}, and \texttt{SUNPHARMA}.

Table~\ref{tab:performance-by-lambda} reports the full set of performance statistics underlying Figure~\ref{fig:performance-comparison} for the Indian equity universe described in the previous section, with LSTM entries averaged across 100 seeds and net metrics computed after 20 bps one-way transaction costs.

\begin{table*}[h]
\centering
\caption{Performance by $\lambda.$}
\label{tab:performance-by-lambda}
\small
\setlength{\tabcolsep}{3.5pt}
\begin{tabular}{@{}llrrrrrrr@{}}
\toprule
$\lambda$ & Model & Gross SR & Net SR & Sortino & Net Vol. & Net MDD & Final wealth & Daily turnover \\
\midrule
0.0 & Markowitz & 0.050 & 0.019 & 0.029 & 13.3\% & -25.0\% & 0.988 & 0.0081 \\0.0 & LSTM & 0.037 & -0.021 & -0.033 & 13.2\% & -24.7\% & 0.978 & 0.0153 \\
\addlinespace[2pt]
0.1 & Markowitz & 0.235 & 0.204 & 0.306 & 13.1\% & -23.7\% & 1.037 & 0.0082 \\
0.1 & LSTM & 0.220 & 0.150 & 0.229 & 13.0\% & -23.1\% & 1.022 & 0.0180 \\
\addlinespace[2pt]
0.2 & Markowitz & 0.670 & 0.638 & 0.949 & 13.0\% & -21.4\% & 1.159 & 0.0082 \\
0.2 & LSTM & 0.639 & 0.562 & 0.839 & 13.3\% & -20.9\% & 1.139 & 0.0200 \\
\addlinespace[2pt]
0.3 & Markowitz & 1.059 & 1.027 & 1.546 & 12.9\% & -19.4\% & 1.279 & 0.0082 \\
0.3 & LSTM & 0.959 & 0.885 & 1.304 & 13.5\% & -19.1\% & 1.243 & 0.0197 \\
\addlinespace[2pt]
0.4 & Markowitz & 1.456 & 1.426 & 2.172 & 13.3\% & -16.9\% & 1.430 & 0.0079 \\
0.4 & LSTM & 1.297 & 1.235 & 1.819 & 13.7\% & -17.2\% & 1.371 & 0.0168 \\
\addlinespace[2pt]
0.5 & Markowitz & 1.801 & 1.774 & 2.734 & 14.1\% & -14.1\% & 1.609 & 0.0075 \\
0.5 & LSTM & 1.733 & 1.691 & 2.580 & 14.0\% & -14.4\% & 1.565 & 0.0118 \\
\addlinespace[2pt]
0.6 & Markowitz & 1.902 & 1.877 & 2.905 & 14.7\% & -13.4\% & 1.688 & 0.0073 \\
0.6 & LSTM & 1.868 & 1.836 & 2.826 & 14.5\% & -13.5\% & 1.659 & 0.0092 \\
\addlinespace[2pt]
0.7 & Markowitz & 1.946 & 1.923 & 2.980 & 15.1\% & -13.1\% & 1.738 & 0.0071 \\
0.7 & LSTM & 1.908 & 1.881 & 2.904 & 14.9\% & -13.2\% & 1.702 & 0.0082 \\
\addlinespace[2pt]
0.8 & Markowitz & 1.953 & 1.931 & 2.992 & 15.4\% & -13.4\% & 1.758 & 0.0069 \\
0.8 & LSTM & 1.924 & 1.898 & 2.931 & 15.2\% & -13.4\% & 1.728 & 0.0078 \\
\addlinespace[2pt]
0.9 & Markowitz & 1.957 & 1.935 & 2.995 & 15.7\% & -13.6\% & 1.777 & 0.0068 \\
0.9 & LSTM & 1.932 & 1.907 & 2.940 & 15.5\% & -13.6\% & 1.750 & 0.0077 \\
\addlinespace[2pt]
1.0 & Markowitz & 1.957 & 1.936 & 2.987 & 16.0\% & -13.8\% & 1.796 & 0.0066 \\
1.0 & LSTM & 1.936 & 1.912 & 2.937 & 15.7\% & -13.8\% & 1.769 & 0.0074 \\
\bottomrule
\end{tabular}
\end{table*}

\begin{figure*}[p]
    \centering
    \includegraphics[width=0.94\textwidth]
    {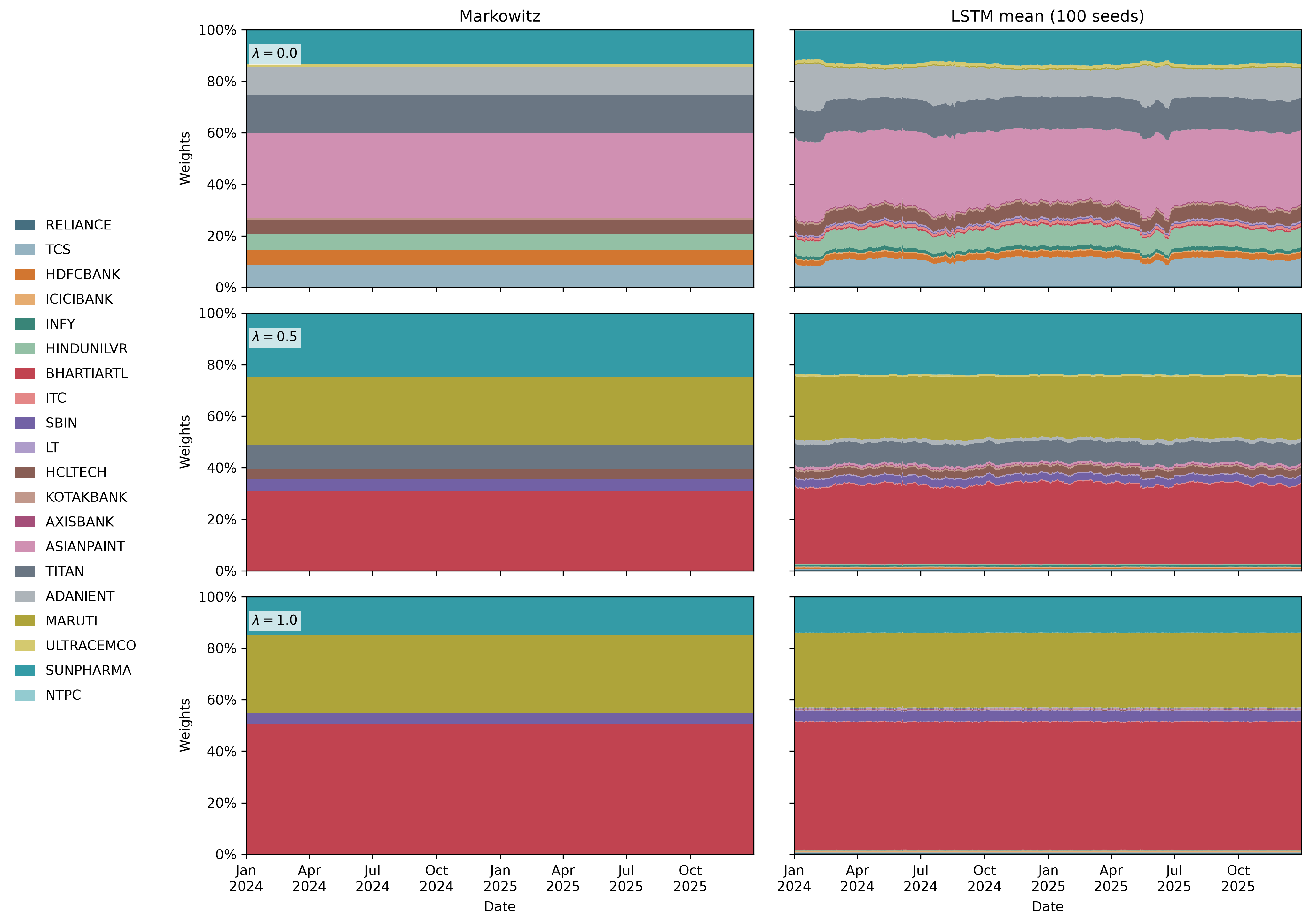}

    \caption{
    Portfolio allocations over 2024--2025 for
    $\lambda=0$, $0.5$, and $1$.
    The left column shows the Markowitz target allocations, while the
    right column shows the mean LSTM allocations across 100 random seeds.
    }
    \label{fig:portfolio-composition}
\end{figure*}

\end{document}